\documentclass[twocolumn]{autart}
\usepackage{euscript,yfonts,dsfont,textcomp,amsmath,amssymb,amstext,parskip,url,graphicx,tikz,float,subcaption}
\usepackage{graphics}
\usepackage{enumitem}

\newenvironment{proof}{\paragraph*{\mdseries \textit{Proof:}}}{\hfill$\square$}

\newcommand{\blue}{\color{black}}

\newcommand{\black}{\color{black}}
\definecolor{gray}{rgb}{0.5,0.5,0.5}

\newcommand{\X}{\mathcal X}
\newcommand{\Y}{\mathcal Y}

\newtheorem{problem}{Problem}
\newtheorem{lemma}{Lemma}

\newtheorem{remark}{Remark}
\newtheorem{proposition}{Proposition} 

\newtheorem{innercustomgeneric}{\customgenericname}
\providecommand{\customgenericname}{}
\newcommand{\newcustomtheorem}[2]{%
  \newenvironment{#1}[1]
  {%
   \renewcommand\customgenericname{#2}%
   \renewcommand\theinnercustomgeneric{##1}%
   \innercustomgeneric
  }
  {\endinnercustomgeneric}
}

\newcustomtheorem{customthm}{Problem}
\newcustomtheorem{customlemma}{Lemma}

\def\spacingset#1{\def\baselinestretch{#1}\small\normalsize}
\begin{document}

\begin{frontmatter}

\title{Monge-Kantorovich Optimal Transport\\
Through Constrictions and Flow-rate Constraints\thanksref{footnoteinfo}}

\thanks[footnoteinfo]{This paper was not presented at any IFAC 
meeting. Corresponding author Tryphon T.\ Georgiou. Tel. 949-824-9966.}

\author[1]{Anqi Dong}\ead{anqid2@uci.edu},    
\author[2]{Arthur Stephanovitch}\ead{arthur.stephanovitch@ens-paris-saclay.fr},  
\author[1]{Tryphon T.\ Georgiou}\ead{tryphon@uci.edu}

\address[1]{Department of Mechanical and Aerospace Engineering, University of California, Irvine, CA 92697, USA}      
\address[2]{\'Ecole normale sup\'erieure Paris-Saclay, France}

\begin{keyword}                                
Optimal mass transport, Flow-rate constraints, Optimal Control, Transportation control, Flow control.
\end{keyword}                            
\begin{abstract}
We consider the problem to transport resources/mass while abiding by constraints on the flow through constrictions along their path between specified terminal distributions. Constrictions, conceptualized as toll stations at specified points, limit the flow-rate across.
We quantify flow-rate constraints via a bound on a sought probability density of the times that mass-elements cross toll stations and cast the transportation scheduling in a Kantorovich-type of formalism. Recent work by our team focused on the existence of Monge maps for similarly constrained transport minimizing average kinetic energy. The present formulation in this paper, besides being substantially more general, is cast as a (generalized) multi-marginal transport problem -- a problem of considerable interest in modern-day machine learning literature and motivated extensive computational analyses.
An enabling feature of our formalism is the representation of an average quadratic cost on the speed of transport as a convex constraint that involves crossing times.
\end{abstract}

\end{frontmatter}

\spacingset{.96}
\section{Introduction}\label{sec:introduction}
Recent years have witnessed the rapid development of the Monge-Kantorovich Optimal Mass Transport (OMT) theory, driven by a plethora of applications from computer vision to geophysics and from thermodynamics to economics and many more.
The basic problem that dates back to the $18^{\rm th}$ century is that of optimizing with respect to suitable transportation costs the schedule for transporting mass to meet supply and demand. This problem has been of such importance that brought the Nobel prize to Leonid Kantorovich who in 1942 devised duality theory and linear programming for its solution. Starting at the waning years of the $20^{\rm th}$ century, contributions by Brenier, McCann, Otto, Gangbo, Evans, Villani, and many others, sparked a new phase of rapid development with significant inroads of OMT as an enabling tool in mathematics and physics \cite{villani2021topics}.

Imposition of physical constraints along the transport has been studied by many authors in various contexts. In particular, moment-type constraints were considered 
in \cite[Section 4.6.3]{rachev1998mass}, \cite{ekren2018constrained}, a path-dependent cost has been studied in
\cite{carlier2008optimal}, \cite[Section 4]{santambrogio2015optimal} as a way to deal with congestion, while bounds on probability mass-densities have considered
in \cite{korman2013insights,rachev2006mass,gladbach2022limits} for similar purposes.
In the present work we revisit a formulation in our earlier work \cite{stephanovitch2022optimal}
that aims to limit flow-rate. Specifically, we seek to transport mass through constriction points/tolls abiding by flow-rate bounds. This earlier work proved the existence and uniqueness of flux-limited Monge maps (scheduling maps) that effect transport while minimizing the $\mathcal W_2$ Wasserstein length of trajectories. In the present work, we take an alternative view by casting the problem in a Kantorovich-type in the form of multi-marginal optimization. We note that, while multi-marginal problems have been studied for a while \cite{pass2015multi,kim2014general}, issues related to computation for large size problems continue to be of great interest to the present day -- we refer to \cite{haasler2021multimarginal,haasler2021scalable,haasler2021multi} for structured multi-marginal problems and to \cite{peyre2019computational} for the computational aspects of OMT in general.

The key idea in the present work for dealing with flow-rate constraints is to seek a distribution for the times of crossing the toll-stations, suitably bounded to ensure meeting the constraints. Our choice for a cost to be minimized is the mean-square of the velocity. This choice effectively orders the flow, in that particles do not overtake each other, and can be expressed as a convex functional of the optimization variables.
Moreover, this cost can be conveniently decomposed into a sum of costs, each engaging only space and timing variables that pertain to segments of the transport path. We ensure flux constraints by imposing bounds on the variables representing the time of toll-crossing. Under a fairly general setting our formulation amounts to a multi-marginal OMT problem with the time-crossing marginals as unknown and constrained parameters of the problem.

Below, in Section \ref{sec:OMTthroughtoll} we discuss the formulation of Monge-Kantorovich transport through tolls with flow-rate constraint.
In Section \ref{sec:one-dimensionalOMT} we specialize to measures with support on $\mathbb R$ and we highlight the nature of solutions with representative examples. In Section \ref{sec:Extension}, we discuss two generalizations of the basic problem, the first concerns partial transport through tolls where mass needs only to clear certain tolls on the way to the respective destination, while the second brings in a new dimension to the transport problem by considering arrival and departure times in the formulation.

\section{Monge-Kantorovich transport through tolls}\label{sec:OMTthroughtoll}

We consider the classical optimal mass transport problem with a quadratic cost, albeit with a flow-rate constraint on the flow across tolls at specified locations. 
Specifically, we consider distributions $\mu,\nu$, that are assumed throughout as being probability measures (i.e., with mass normalized to $1$), and seek a transportation plan from one to the other. The ``starting'' distribution $\mu$ represents the mass of a single commodity that needs to be transported accordingly and matches the demand that is specified by $\nu$. Along the way, the mass has to clear tolls abiding by corresponding constraints on throughput.

Let us first review the classical Monge-Kantorovich optimal mass transport with a quadratic cost (on $\mathbb R^n$). The Monge formulation of transport seeks a transportation map $T\;:\;x\mapsto y=T(x)$ so as to minimize {\blue the transportation cost functional}
$$
J(T)=\int_{\mathbb R^n} \|T(x)-x\|^2 \mu(dx),
$$
{\blue over the choice of $T$,} subject to {\blue the transportation map $T$ pushing the starting to the target distribution, i.e., $\mu$ to $\nu$. We denote this by writing} $T_\sharp \mu = \nu$. {\blue As it is standard,} the notation $\sharp$ denotes the ``push forward,'' meaning that for any Borel set $S\subset \mathbb R^n$, $\mu(T^{-1}(S))=\nu(S)$. 

{\blue In general,} such a map $T$ may not always exist. The Kantorovich relaxation seeks instead a measure $\pi$ on the product space $\mathbb R^n\times \mathbb R^n$ so that the marginals on the two components coincide with $\mu,\nu$, while $\pi$ minimizes the functional
$$
J(\pi)=\int_{\mathbb R^n\times \mathbb R^n} \|y-x\|^2 \pi(dx,dy).
$$
{\blue The measure $\pi$ is referred to as coupling of the two marginal distributions}.

Assuming that $\mu$ is absolutely continuous with density $\rho_0$ then the optimal transport map $T$ {\blue always} exists and is unique \cite[Theorem 2.12]{villani2021topics}, and the optimal cost is the minimal of
\begin{equation}\label{eq:Jv} 
J(\partial_t X) := \int_0^1 \int_{\mathbb R^n} \big(\partial_t X_t(x)\big)^2 \rho_0(x)dxdt,
\end{equation}
for $X:[0,1]\times \mathbb R^n \rightarrow \mathbb R^n$ such that\footnote{Following a common abuse of notation, for simplicity, we write interchangeably the density $\rho_0$ and measure $\mu$, and allow the notation to be understood from the context.}
$X_{0\sharp}\rho_0 = \rho_0$ and  $X_{1\sharp} \rho_0 = \nu$.
{\blue That is, the minimization is over velocity fields $\partial_t X$, for flows $X$ that corresponds the starting and ending marginal densities.}
{\blue It turns out that} the optimal flow is effected by 
\begin{align*}
X^{\rm opt}_t(x)= x + t\big(T(x)-x\big),
\end{align*}
for $T$ the optimal transport map. The minimal value of these functionals is designated as the Wasserstein distance $\mathcal{W}_2^2(\mu,\nu)$ \cite{villani2009optimal}. 

One observes that, in \eqref{eq:Jv}, what is actually being minimized is the average kinetic energy (modulo a factor of $2$), as the transport takes place over the interval $[0,1]$.
In general, having $d\mu=\rho_0dx$ and carrying out the transport over the interval $[0,t_f]$,
\begin{align*} 
\mathcal{W}_2^2(\mu,\nu) & {\blue = J(\partial_t X_t^{\rm opt})}\\
&= \int_0^{t_f} t_f\int_{\mathbb R^n} {\blue \big(\partial_t X^{\rm opt}_t(x)\big)^2} \rho_0(x)dxdt\\
& = \int_0^{t_f} \int_{\mathbb R^n} \frac{\|T(x)-x\|^2}{t_f}\rho_0(x)dxdt
\end{align*}   
since ${\blue \partial_t X^{\rm opt}_t(x)= \big(T(x)-x\big)/t_f}$, recovering $J(T)$ for the optimal transport map. {\blue Most importantly, {\em the average kinetic energy can also be written}}, for the relaxed Kantorovich formulation,
as
\[
\mathcal{W}_2^2(\mu,\nu) =\int_0^{t_f} \int_{\mathbb R^n\times \mathbb R^n} \frac{\|y-x\|^2}{t_f}\pi(dx,dy) dt.
\]
The derivation is immediate since $\int_0^{t_f}dt=t_f$. The expression helps highlight the average kinetic energy as the minimal of the convex cost $\|y-x\|^2/t_f$, when transporting mass from $x$ to $y$ over a time interval of duration $t_f$.

We are now in a position to {\em \blue formulate the analogue of the Monge-Kantorovich problem for the case where the mass needs to clear tolls}, initially a single toll, at specified locations and with a bound on the flow-rate. 

We first assume that a single toll is located at $\xi$ and that mass flowing through cannot exceed a given flow-rate $r$. {\blue Thus, the flow-rate through the toll must satisfy
\begin{align*}
\sigma(t)\leq r. 
\end{align*}
The flow-rate represents the mass density over time as the mass is transported through the toll, i.e., $\sigma(t)dt$ represents mass that clears the toll in the interval $[t,t+dt]$.

The insight that allows us to formulate optimal transport through toll(s) as a multi-marginal optimal transport problem is to view the {\em as-yet-undermined} mass density $\sigma(t)$ as a time marginal distribution, together with the specified spatial initial and final marginals $\mu(dx)$ and $\nu(dy)$.
Thereby, the Kantorovich formulation of the problem seeks a coupling $\pi(dx,dy,dt)$ between the two given marginals $\mu(dx)$, $\nu(dy)$, and the sought marginal $\sigma(t)dt$, that specifies the portion of mass at $[x,x+dx]$, heading towards $[y,y+dy]$, is transported through the toll in the interval $[t,t+dt]$.

In light of the coupling between a starting point $x$ at time $\tau=0$, location of the toll $\xi$ that mass is to be transported at time $\tau=t$, and terminal destination $y$ to arrive at $\tau=t_f$, the average kinetic energy\footnote{The average kinetic energy is often referred to as {\em action integral} in physics.} over the segment $\tau\in[0,t]$ is minimized when particles travel at constant velocity $(\xi-x)/t$, and thus, equal to
\[
\int_0^t \|(\xi-x)/t\|^2 d\tau = \|\xi-x\|^2/t.
\]
Likewise, the average kinetic energy over the remaining interval $[t,t_f]$ is $\|y-\xi\|^2/(t_f-t)$.
}
Thus, we arrive at the formulation of our first problem.

\begin{problem}\label{prob:problem1}
Given probability measures $\mu,\nu$ and $r>0$, determine a probability measure $\pi(dx,dy,dt)$ on $\mathbb R^n\times \mathbb R^n\times [0,t_f]$ that minimizes
\begin{align}
    \iiint_{x,y,t} \left( 
    \frac{\|\xi-x\|^2}{t}+\frac{\|y-\xi\|^2}{t_f-t}
    \right)  \pi(dx,dy,dt),
\label{eq:costProblem1}
 \end{align}  
subject to the marginals
$\iint_{y,t}\pi(dx,dy,dt) =\mu(dx)$ and
$\iint_{x,t}\pi(dx,dy,dt) =\nu(dy)$, and the flow-rate constraint $\iint_{x,y}\pi(dx,dy,dt) \leq r dt$. 
\end{problem}

Throughout, notation as in $\iint_{x,y}\pi(dx,dy,dt)$, indicates integration over the space of the variables that are subscribed to the integral.

{\blue Note that} the above formulation allows mass that is initially concentrated  $x$ and is heading towards the same terminal destination $y$, to split and transport over different paths clearing the toll at different times $t$, as prescribed by the coupling, in order to abide by the imposed bound on flow-rate.
{\blue This will necessarily be the case when $\mu$ assigns finite mass at a point (contains Dirac deltas or singular part, in higher dimensions).}

A somewhat more ambitious scheduling may require optimizing the average kinetic energy, from a starting measure $\mu$ to the terminal $\nu$, clearing two or multiple tolls in succession.
For the case of two tolls, define  $\pi(dx,dy,dt_1,dt_2)$ as the coupling of mass at $x$ heading to $y$ that clears the tolls at times $t_1,t_2$, respectively, where $0\leq t_1\leq t_2\leq t_f$. The coupling, as before, is a probability measure that specifies the respective amount of mass that is transported as prescribed.

The problem with two or multiple tolls is analogous to Problem \ref{prob:problem1}. For instance, in the case of two tolls, the cost  $c^{\xi_{1},\xi_{2}}$ to be minimized over the choice of admissible couplings is
\begin{align*}
    \iiiint_{x,y,t_1,t_2} \left( 
    \frac{\|\xi_1-x\|^2}{t_1}+\frac{\|\xi_2-\xi_1\|^2}{t_2-t_1}+\frac{\|y-\xi_2\|^2}{t_f-t_2}
    \right)  \pi.
 \end{align*}
Admissibility of $\pi$ amounts to consistency with the problem data, i.e., it amounts to satisfying the usual marginal constraints on $x$ and $y$, as well as the flow-rate constraints
$\iiint_{x,y,t_2}\pi(dx,dy,dt_1,dt_2) \leq r_1 dt_1$ and
$\iiint_{x,y,t_1}\pi(dx,dy,dt_1,dt_2) \leq r_2 dt_2$.

We remark that flow-rate constraints can be time-varying without any significant overhead in the difficulty of the problem. Specifically, in the condition $\iint_{x,y}\pi(dx,dy,dt) \leq r dt$ the flow-rate bound can be specified by a time-dependent density $r(t)$ that regulates permissible {\em throughput} at different times. Evidently, the bound could also be a measure, but this is deemed of minimal practical relevance and not followed here.

Our first technical result establishes in a straightforward manner existence of solutions.
\begin{proposition}[Existence]\label{prop:prop1} Provided $r t_f>1$,
Problem \ref{prob:problem1} admits a (minimizing) solution $\pi$.
\end{proposition}

\begin{proof}
{\blue Denoting with $\sigma(t)$ the mass density that crosses the toll at time $t$, as before, the transport of the total mass through the toll over the interval $[0,t_f]$ subject to $\sigma(t)\leq r$, requires that
\begin{align*}
    1=\int_{0}^{t_f} \sigma(t) dt \leq \int_{0}^{t_f} r dt=rt_f,
\end{align*}
Thus, $rt_f\geq 1$ is a necessary condition\footnote{\blue The case $rt_f=1$ is only feasible in the non-generic situation where $\max {\rm Support}(\mu)=0=\min {\rm Support}(\nu)$. In general, when e.g.,  $\max {\rm Supp}(\mu)<0$, the ``rightmost'' mass on the support of $\mu$ must be transported with infinite velocity 
so as to allow $\sigma(t)=1/t_f$ over $[0,t_f]$ for the total mass to have enough time to be transported through. Such limiting cases are non-physical, leading to diverging transportation cost, and thereby excluded.}.}
The space of admissible measures $\pi$ in Problem  \ref{prob:problem1}, i.e.,
\begin{align*}
\Pi= \bigg\{&\pi \in \mathcal{P}\big(\mathbb{R}^n,\mathbb{R}^n,[0,t_f]\big) \mid \iint_{y,t}\pi(dx,dy,dt) =\mu(dx),\; \\
& \hspace*{-20pt}\iint_{x,t}\pi(dx,dy,dt) =\nu(dy),
\iint_{x,y}\pi(dx,dy,dt) \leq  r dt \bigg\}
\end{align*}
is non-empty. {\blue It is also compact for the weak topology as it is tight and closed for the narrow convergence. Indeed, the set of coupling measures between a finite number of probability measures is tight \cite[Theorem 1.4]{santambrogio2015optimal}, and each of  the three constraints in $\Pi$ is closed for the narrow convergence. Then, as the cost function (integrand in \eqref{eq:costProblem1}) is lower semi-continuous, we have the existence of a minimizer. To see that $\Pi$ is non-empty, we postulate a uniform distribution on the crossing times, $u=(1/t_f)dt$, and couplings $\tilde\pi_{xt},\tilde\pi_{ty}$ that are consistent with the marginals $(\mu,u)$ and $(u,\nu)$, respectively. Then, the existence of an element $\tilde \pi\in\Pi$ with marginals
\begin{align*}
\int_{y} \tilde \pi(dx,dy,dt) = \tilde\pi_{xt}, \ \mbox{and} \ \int_{x} \tilde \pi(dx,dy,dt)  = \tilde\pi_{yt},
\end{align*}
is} guaranteed by the gluing lemma \cite[page 11]{villani2009optimal}.
\end{proof}

Note that the convexity of the cost is not used in the proposition. We next consider whether the optimal solution is unique. We first discuss the special case where $n=1$. Moreover, for this case where locations on the underlying space can be ordered (e.g., from left to right), we assume that the toll sits between the two distributions $\mu,\nu$, specifically that the support of $\mu$ is to the left of $\xi$ and that the support of $\nu$ is to the right. Without loss of generality, we let $\xi=0$ and we thus consider the following problem.

\begin{customthm}{$\mathbf 1^\prime$}[Simplification]\label{prob:problem1prime}
We consider probability measures $\mu,\nu$ on $\mathbb R$, with support on $[-M,0)$ and $(0,M]$ for sufficiently large $M$, respectively. Let $r,t_f>0$ such that $rt_f>1$.
Determine a probability measure $\pi(dx,dy,dt)$ as the minimizer of
\begin{align*}
    \iiint_{x,y,t} \left( 
    \frac{x^2}{t}+\frac{y^2}{t_f-t}
    \right)  \pi(dx,dy,dt)
 \end{align*}  
 subject to the flow-rate constraint $\iint_{x,y}\pi(dx,dy,dt) \leq r dt$, and the marginals
 $\iint_{y,t}\pi(dx,dy,dt) =\mu(dx)$ and
 $\iint_{x,t}\pi(dx,dy,dt) =\nu(dy)$.
\end{customthm}

We begin with a technical lemma that establishes a correspondence between the time and location of mass as this is transported past the toll. A schematic that exemplifies the statement of the lemma below is shown in Fig.~\ref{fig:diracillustration}.

\begin{lemma}[Monotonicity]\label{lemma:keylemma}
Let $c_{xt}:=\frac{x^2}{t}$, with $(x,t) \in[-M,0) \times (0,t_f]$, and $\mu,\sigma$ measures with support on $[-M,0)$ and $(0,t_f]$, respectively, with $\sigma(t)dt$ absolutely continuous with respect to the Lebesgue measure. The minimizer of the Kantorovich problem
\begin{align*}
    \min_{\pi}\iint c_{xt}\pi,
\end{align*}
where $\pi$ represents a coupling of the two marginals $\mu(dx),\sigma(t)dt$,
is unique with support on the graph of a non-increasing function $T^\X(t)$.
\end{lemma}

\begin{proof}
We first observe that for any two pairs $x,x'\in [-M,0)$ and $t,t'\in(0,t_f]$ for which {\blue $0>x>x'\geq-M$} and {\blue $t_f \geq t>t'>0$}, it holds that
\begin{align*}
\frac{x^2}{t}+\frac{(x')^2}{t'} > \frac{x^2}{t'}+\frac{(x')^2}{t}.    
\end{align*}
The ordering in this inequality characterizes $c_{xt}$ as being quasi-monotone, in the language of \cite{cambanis1976inequalities}, see also \cite[Section 3.1]{rachev1998mass}. 

It follows from \cite[Theorem 2.18, and Remark 2.19]{villani2021topics} that the optimal coupling $\pi$ exists and is given by the monotone rearrangement of $\mu,\sigma$, that is, for a suitable function $T^\X(t)$
\begin{align*}
{\blue \int_{T^\X(t)}^{0}\mu(dx) =  \int_0^t\sigma(s)ds.}   
\end{align*}
Since $T^\X(t)<0$, it is non-increasing (and is constant on time-intervals that correspond to possible Dirac components of $\mu$). This completes the proof.
\end{proof}

{\blue For similar reasons, the cost $c_{yt}$ (with  $t\in(0,t_f]$ and $y\in(0,M]$, and $M$ as in Problem \ref{prob:problem1prime}) is quasi-monotone. Hence, once again, for a suitable function $T^\Y(t)$,
\begin{align*} 
    \int_{T^\Y(t)}^{M}\!\!\nu(dy) =  \int_0^t\sigma(s)ds,
\end{align*}
for $T^\Y(t)>0$, so that $T^\Y(t)$ is non-increasing.
In light of the monotonicity of $T^\X(t)$ and $T^\Y(t)$, we establish the following proposition. 
}

\begin{proposition}[Uniqueness]\label{prop:prop2} Under the assumptions of Problem \ref{prob:problem1prime} the minimizer is unique. Moreover, there are functions $T^\X(t),T^\Y(t)$ are monotonically non-increasing such that
\begin{align*}
\pi=(T^\X,T^\Y,{\rm Id})_{\sharp \sigma},    
\end{align*}
where ${\rm Id}(t)=t$ is the identity map and $\sigma(t)dt$ is an absolutely continuous measure on $[0,t_f]$ with $\sigma(t)\leq  r$.
\end{proposition}

\begin{proof}
Let $\pi$ be a minimizer as claimed in Proposition \ref{prop:prop1}, and let
$\pi_{xt}:=\int_y \pi$, $\pi_{yt}:=\int_x \pi$, and
$\sigma:=\iint_{xy} \pi$. Since
\[
\iiint_{x,y,t} c\pi=\iint_{xt}c_{xt}\pi_{xt}+
\iint_{yt}c_{yt}\pi_{yt},
\]
$\pi_{xt}$ is a minimizer of $\iint_{xt}c_{xt}\pi_{xt}$, and the same applies to $\pi_{yt}$. If this was not the case, there would be couplings $\hat\pi_{xt},\hat\pi_{yt}$ with strictly lower costs $\iint_{xt}c_{xt}\hat\pi_{xt}$, and $\iint_{yt}c_{yt}\hat\pi_{yt}$. These two couplings share the same marginal on the $t$-axis, namely,
\[
\int_x\hat\pi_{xt}(dx,dt)=\int_y\hat\pi_{yt}(dy,dt)=\sigma(t)dt.
\]
Then, by the gluing lemma \cite[page 11]{villani2009optimal}, there is a coupling $\hat \pi$ on $\mathbb R\times\mathbb R{\blue\times[0,t_f]}$ that agrees with the given marginals and has a lower cost. Thus, both $\pi_{xt},\pi_{yt}$ are optimal for the respective problems. We next argue that $\sigma$ is unique, and therefore, the conclusion follows by Lemma \ref{lemma:keylemma}.

To establish the uniqueness of the density on the $t$-axis, assume that there are two different minimizers $\pi^a$ and $\pi^b$ to start with.
Then, as above, each gives rise to a density on the $t$-axis, $\sigma^a(t)$ and $\sigma^b(t)$, respectively, as well as corresponding marginals and maps $(T^{\X,a},T^{\Y,a})$ and $( T^{\X,b},T^{\Y,b})$. Since both $\sigma^a(t)$ and $\sigma^b(t)$ satisfy the constraint of being $\leq r$, so does any convex linear combination, say $\bar \sigma=\frac12 \sigma^a+\frac12 \sigma^b$, and the convex combination $\bar \pi=\frac12 \pi^a+\frac12 \pi^b$ is also optimal. But then, the marginal 
\begin{align*}
\bar \pi_{xt}=  (T^{\X,a},{\rm Id})_{\sharp \frac12 \sigma^a} +
(\hat T^{\X,b},{\rm Id})_{\sharp \frac12 \sigma^b}
\end{align*}
is supported on a set that is not the graph of a function\footnote{\blue Note that $\bar \pi_{xt}$ satisfies the marginal constraints since, by virtue of $T^\X_{\sharp \sigma^i} = \mu$, for $i\in\{a,b\}$,  $T^\X_{\sharp \frac12 \sigma^a} + T^\X_{\sharp \frac12 \sigma^b}=\mu$. A similar statement holds for the  coupling $\bar \pi_{yt}$.}, unless of course $T^\X=T^{\X,a}=T^{\X,b}$. 
{\blue 
But if $\bar \pi_{xt}$ is not supported on a graph of a function, there exists a more ``economical'' coupling with strictly lower cost, obtained by 
 monotone rearrangement of $\bar \pi_{xt}$. A similar statement holds for $\bar \pi_{yt}$. This contradicts the nonuniqueness and completes the proof.}
\end{proof}

In the setting of Problem \ref{prob:problem1}, when $n>1$,  the transport cost of all mass that resides at a distance  $d=\|x-\xi\|$ is the same. Thus, the problem to transport through the toll cannot distinguish equidistant points from the toll.
In this case where the distributions, $\mu,\nu$ sit in $\mathbb{R}^n$ for $n\geq 1$, Proposition \ref{prop:prop1} and \ref{prop:prop2} can be readily extended as the cost in Problem \ref{prob:problem1} only depends on the distance of the points to $\xi$. Indeed, the problem is equivalent to solving the 1-dimensional problem between $\widetilde{\mu},\widetilde{\nu}$ the  measures such that for all $ A\subset \mathbb{R}$ measurable set\footnote{The notation $\mathds{1}_{S}(x)$, or $\mathds{1}_{S}$ for simplicity,  signifies the indicator function that takes the value $1$ when $x\in S$ and zero otherwise.},
$$
\widetilde{\mu}(A)=\int_{\mathbb{R}^n} \mathds{1}_{\{||x-\xi||\in A\}}\mu(dx).  
$$

Therefore we have existence of a unique solution $\widetilde{\pi}=(T^\X,T^\Y,{\rm Id})_{\sharp \sigma}$ to the 1-dimensional problem which gives rise to solutions $\pi$ to the n-dimensional problem in the following way:
For $\mu_{\tilde x},\nu_{\tilde y}$ the disintegrated measures \cite{chang1997conditioning} such that
$$\mu(dx) = \int \mu_{\tilde x}(dx)\widetilde{\mu}(d\tilde x)$$
$$\nu(dy) = \int \nu_{\tilde y}(dy)\widetilde{\nu}(d{\tilde y}),$$
the solutions $\pi$ will be of the form
$$\pi(dx,dy)=\int \pi_{t}(dx,dy) \sigma(t)dt$$
for $\pi_{t}$ any coupling\footnote{The coupling between $\mu_{T^\X(t)}$ and $\nu_{T^\Y(t)}$ cannot be specified from the problem setting, since it does not affect the cost.} measure between $\mu_{T^\X(t)}$ and $\nu_{T^\Y(t)}$.

\begin{remark}[Generalization]\label{remark:general}
A further interesting generalization is when the toll through which the mass is to be transported is no longer a point but a set $\mathcal T \subset \mathbb{R}^n$, typically a curve or a manifold of higher dimension, 
with a Hausdorff measure $\mathcal{H}(dz)$ integrating to $1$. \black The transport problem for such a situation becomes one of minimizing 
$$
\inf_{\pi\in \Pi}\int \left(\frac{||x-z||^2}{t}+\frac{||y-z||^2}{t_f-t} \right)\pi(dx,dy,dt,dz),
$$
over couplings in
\begin{align*}
\Pi=& \bigg\{\pi \in \mathcal{P}(\mathbb{R}^n,\mathbb{R}^n,[0,t_f],\mathcal T) \mid \iiint_{y,t,z}\pi =\mu(dx),\; \\
& \hspace*{-10pt}\iiint_{x,t,z}\pi =\nu(dy),
\iint_{x,y}\pi \leq  r dt  \otimes \mathcal{H}(dz)\bigg \}.\end{align*}
The term to the right of the last inequality, representing the (normalized) Hausdorff measure of $\mathcal T$ can be further suitably modified to account for preference/ease of transporting through specific portions of the set $\mathcal T$. Physically such a problem may model flow through media, where $\mathcal T$ represents porous section that the mass must go through, from source $\mu$ to destination $\nu$.
Developing theory for this generality is beyond the scope of the current paper. $\Box$
\end{remark}

\section{Case studies: transport through tolls in 1D}\label{sec:one-dimensionalOMT}
We now present case studies that help visualize the general scheme for Monge-Kantorovich transport through tolls in $\mathbb R$, that is, in dimension $1$. Since the formalism in Section \ref{sec:OMTthroughtoll} casts the problem as a multi-marginal one, the coupling with marginal in 1D is already a measure in $\mathbb R^3$, with one of the axes the time that mass crosses the toll. The computational aspects and the code using the optimization toolbox-CVX \cite{grant2014cvx} to conduct all the experiments can be found at \url{https://github.com/dytroshut/OMT-with-Flux-rate-Constraint}.

We discuss four examples that help visualize the effect flow-rate constraints and the nature and support of the transportation coupling $\pi$.

Our first example is displayed in Fig.~\ref{fig:diracillustration}.
The source distribution is Dirac (i.e., concentrated at one point) located at a point $x < 0$, the toll is located at $\xi=0$, and the terminal distribution is absolutely continuous with respect to the Lebesgue with the density that has support on $\{y\in \mathbb R \mid y\geq 0\}$. The figure highlights the maps $T^\X$ and $T^\Y$ that couple $(x,t)$ and $(y,t)$, respectively, where $t$ denotes the time that mass originally at $x$ crosses the toll on its way to location $y$. Both maps are monotonically non-increasing.

\begin{figure}[htb]
    \centering
    \includegraphics[width=6.5cm]{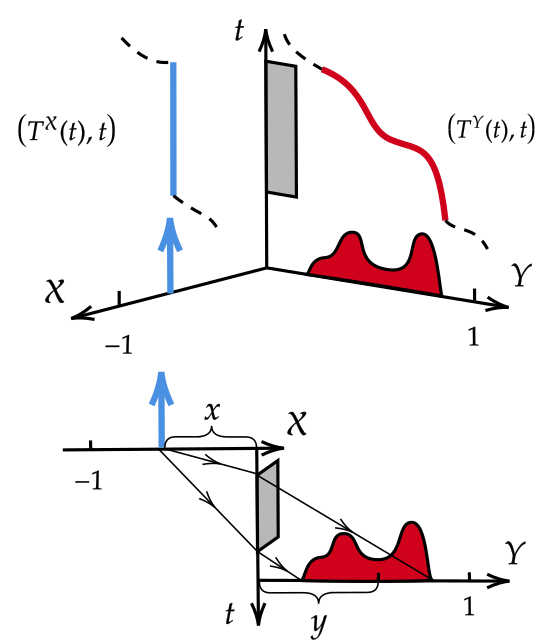}
    \caption{Illustration of solution to Problem \ref{prob:problem1}: the maps $T^\X(t)$, $T^\Y(t)$ are monotonically non-increasing, the $t$-marginal density $\sigma(t)$ is bounded by $r$.}
    \label{fig:diracillustration}
\end{figure}

In our second example, in Fig.~\ref{fig:Randonetoll},
the marginals are Gaussian mixtures.
The flow is visualized via shadowing the paths. Three instances of different bounds on flow-rate are depicted, highlighting how the bound affects the flow and the distribution of crossing times.

\begin{figure}[!htb]
\centering 
\begin{subfigure}{0.75\columnwidth}
\centering
\includegraphics[width=7cm]{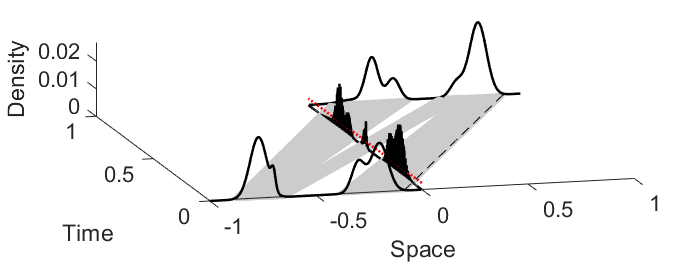}
\caption{$r=\infty$ (no flow-rate constraint)}
\label{fig:Rand2none}
\end{subfigure}
\medskip
\begin{subfigure}{0.75\columnwidth}
\centering
\includegraphics[width=7cm]{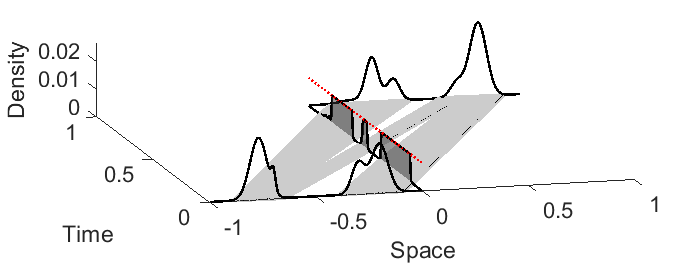}
\caption{$r=2$}
\label{fig:Rand2cap2}
\end{subfigure}
\medskip
\begin{subfigure}{0.75\columnwidth}
\centering
\includegraphics[width=7cm]{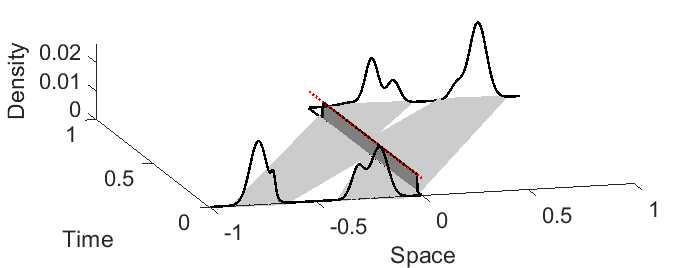}
\caption{$r=1.2$}
\label{fig:Rand12}
\end{subfigure}
\caption{Case (a), having no flow-rate constraint at the toll, corresponds to standard Monge-Kantorovich transport with particles moving at constant speeds depending on origin/destination. Cases (b,c) depict the situation where a flow-rate bound at the toll necessitates that mass is transported with different speeds at the two sides of the toll at $\xi=0$, so as to meet the imposed bound on the $t$-marginal.
}
\label{fig:Randonetoll}
\end{figure}

Our third example in Fig.~\ref{fig:3dcoupling1} helps visualize the coupling in $\mathbb R^3$ between $(x,y,t)$, for smooth marginals; the coupling is supported on a curve. In general, the mass on this curve is not uniform and the density is depicted with circles of suitable radius around corresponding points on the curve.
Couplings for two different values for acceptable flow-rate are shown ($r\in\{1.5, 3\}$) and it is seen that as $r$ is decreased the $t$-density tends to become more uniform.

\begin{figure}[htb]
\begin{subfigure}{1\columnwidth}
\centering
\includegraphics[width=6cm]{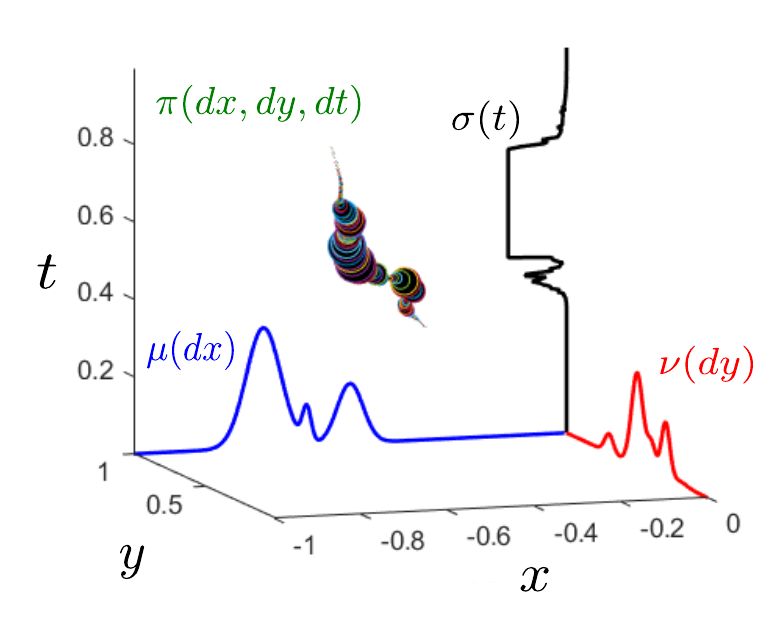}
\caption{$r = 3$}
\end{subfigure}
\medskip
\begin{subfigure}{1\columnwidth}
\centering
\includegraphics[width=6cm]{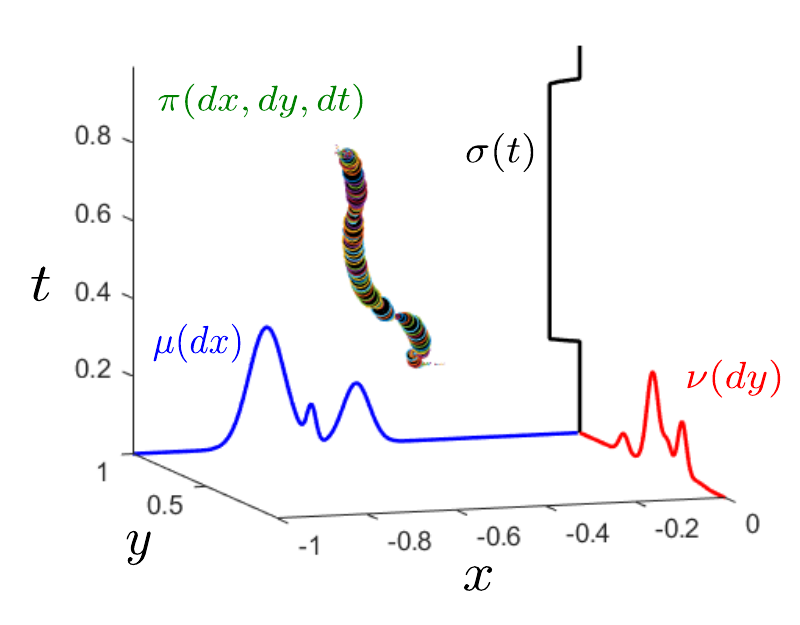}
\caption{$r=1.5$}
\end{subfigure}
\caption{Support of the coupling measure $\pi(dx,dy,dt)$ on a curve in $\mathbb R^3$; the density is depicted by circles of size proportional to magnitude.}
\label{fig:3dcoupling1}
\end{figure}

Our fourth example is drawn with a sketch in Fig.~\ref{fig:twotoll1}, and a simulation of transport between two Gaussian mixtures, through two tolls, in
Fig.~\ref{fig:2successivetolls}. The tolls are positioned at $\xi_1=-0.4$ and $\xi_2=0.4$. The density of the respective times of crossing, $t_1$ and $t_2$, are bounded by $r_1=1.5$ and $r_2=3$, respectively.

\begin{figure}[htb]
\begin{subfigure}{1\columnwidth}
\centering
\includegraphics[width=6.5cm]{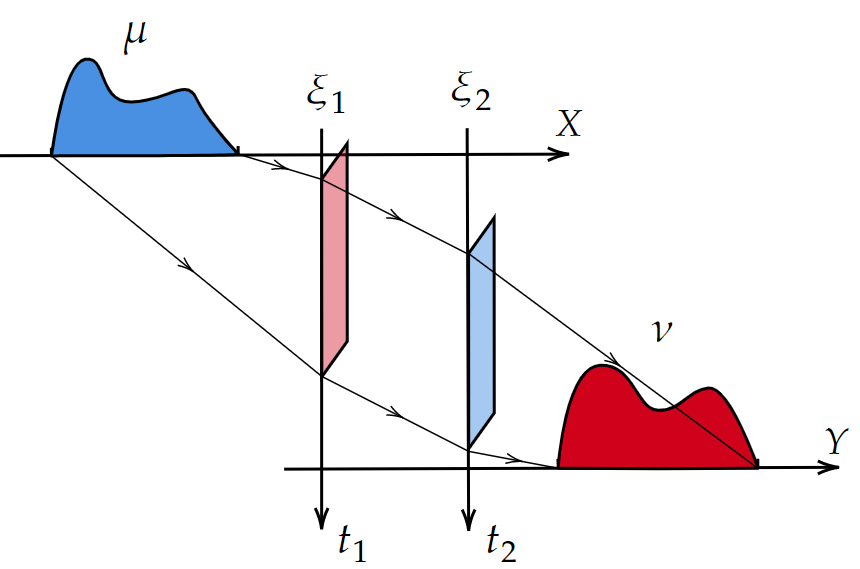}
\caption{Schematic of transport through two tolls at $\xi_1, \xi_2$.}
\label{fig:twotoll1}
\end{subfigure}
\medskip
\begin{subfigure}{1\columnwidth}
\centering
\includegraphics[width=7cm]{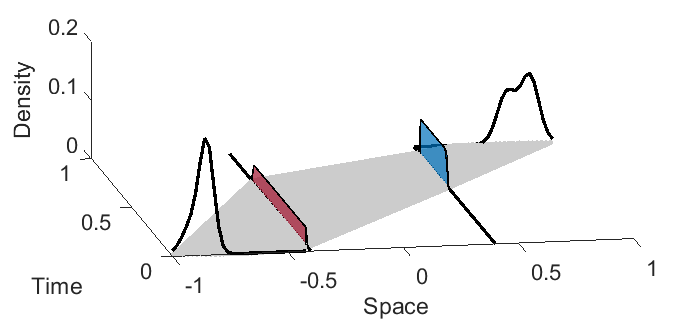}
\caption{Transport of Gaussian mixtures through two tolls with density-bounds  $r_1=1.5$ and $r_2=3$, at crossing times $t_1$ and $t_2$, respectively.}
\label{fig:2successivetolls}
\end{subfigure}
\caption{Schematic and simulation of $2$-toll transport plans.
}
\end{figure}

\section{Concluding Remarks}\label{sec:Extension}

The basic idea presented in this paper for dealing with flow-rate constraints has been to introduce a time-variable for when mass transits certain locations.
{\blue Then, the density of the corresponding marginal distribution, quantifies the amount of mass clearing the toll over a time interval, hence, flow-rate. Such a marginal distribution, as yet to be determined, represents a design parameters to be specified so as to meet flow-rate constraints.} Thereby, such problems can be cast in the form of {\em multi-marginal optimization}.

{\blue We note that the present work builds on, and  extends our earlier study \cite{stephanovitch2022optimal}, where under strong regularity assumptions on the marginals for the supply and demand, we developed a Benamou-Brenier approach for Monge transport through a {\em single} toll.  In contrast to this earlier work, the present formulation allows dealing with more general measures and multiple tolls, and in addition, it casts the problem as a linear program and allows efficient approximation using e.g., entropic regularization as in other timely multi-marginal optimization formulations \cite{haasler2021multimarginal,haasler2021multi}.

We conclude by showing how the basic framework of utilizing marginals to quantify timing information, applies to transportation problems with more complicated structure. Specifically we discuss two cases. First we explain the case where portion of the mass is not constrained to clear the toll, and second, a case of how ordering of arrival and departure times can be incorporated in the same framework.}

\begin{figure}[H]
\centering
\includegraphics[width=6.7cm]{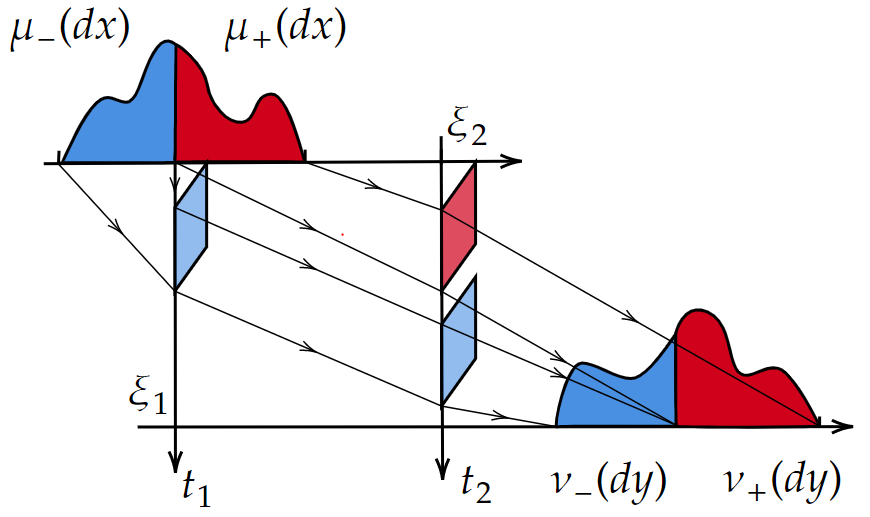}
\caption{Illustration of the optimal mass transport through two successive tolls with separating mass.}
\label{fig:twotoll2}
\end{figure}
\noindent 

For our first case, consider the schematic in Fig.~\ref{fig:twotoll2} where a fraction of a source distribution $\mu$ needs to clear two tolls, while the remaining fraction only needs to clear a single toll, due to its relative position with respect to the toll.
For instance, in this one-dimensional schematic, the source distribution $\mu$ is split by the toll at $\xi_{1}$ into $\mu_{+}:=\mu(x>\xi_{1})$ and $\mu_{-}:=\mu(x<\xi_{1})$. Accordingly, the target distribution $\nu$ needs to be split into $\nu_{+/-}$ corresponding to the two masses, $m_{+/-}$, so that 
\begin{align*}
    T_{\sharp \mu} = T^{\xi_{1},\xi_{2}}_{\sharp \mu_{-}} + T^{\xi_{2}}_{\sharp \mu_{+}} = \nu_{-} + \nu_{+},
\end{align*}
where $T^{\xi_{1},\xi_{2}}$ transports through two tolls whereas $T^{\xi_{2}}$ transports through a single toll.
{\blue With the flow-rate constraints on crossing time marginals, the Kantorovich formulation process is exactly as before, {\em provided the fractions of corresponding masses can be delineated}. For one-dimensional distributions, specifying the corresponding portions is straightforward.} In higher spatial dimensions, when there is no clear separation as in the one-dimensional schematic, the problem of selecting ``what fraction of mass needs to clear what toll'' is coupled to the optimization problem and has a combinatorial nature.

{\blue For our second case, we bring in timing} to prioritize departure and arrival,
so as to meet objectives and possibly mediate congestion along the flow.
To see this, we briefly discuss how to modify the standard Monge-Kantorovich setting in which all particles/mass transport at constant (that depends on the particle) speed along geodesics from source $x$ to destination $y$
according to the McCann flow
\begin{align*}
    \rho_{t} = \big[(1-t){\rm Id}(x)+tT^\Y(x)\big]\sharp \mu,
\end{align*}
over the window $t\in[0,1]$ with transport map $T^\Y(x)$. 
For simplicity we retain arrival time $t_a=1$, i.e., fixed, and only allow the departure time $t_d$ to vary. {\blue The marginal distribution of $t_d$ now represents a design parameter.} 
The formalism, once again, seeks a coupling measure $\pi$ that satisfies constraints and marginals. For simplicity, we introduce the Monge transport map $T^\Y(x)$, to specify the source-destination pairing. Then, the coupling of the variables $x,y,t_d$ gives that $\pi=({\rm Id},T^{t_d},T^\Y)_{\sharp \mu}$. And, if
$\tau(t,x)$ denotes the portion of time that a particle at $x$ is ``on the move,'' i.e., 
$\tau(t,x) = 
(t-t_d(x))\mathds{1}_{\{t > t_{d}(x)\}}$,
the McCann's displacement reads
\begin{align*}
    \rho_{t} = \big[(1-\tau(t,x)){\rm Id}(x)+\tau(t,x)T^\Y(x)\big]\sharp \mu.
\end{align*}
Figure.~\ref{fig:DepArrGaussian} shows an example where both
departure and arrival times are variables {\blue ($t_d$ and $t_a$, respectively)}, with marginals selected to minimize a cost functional of the form $\iiiint_{x,y,t_a,t_d} c(x,y,t_a,t_d) \pi(dx,dy,dt_a,dt_d)$, with cost $c(x,y,t_a,t_d) =t_d/x^2 + (y-x)^2 - t_a/y^2$. {\blue This choice} ensures a natural order in departure and arrival\footnote{The agents/particles that are closer to the destination leave first.}.

\begin{figure}[htb]
\centering
\includegraphics[width=8cm]{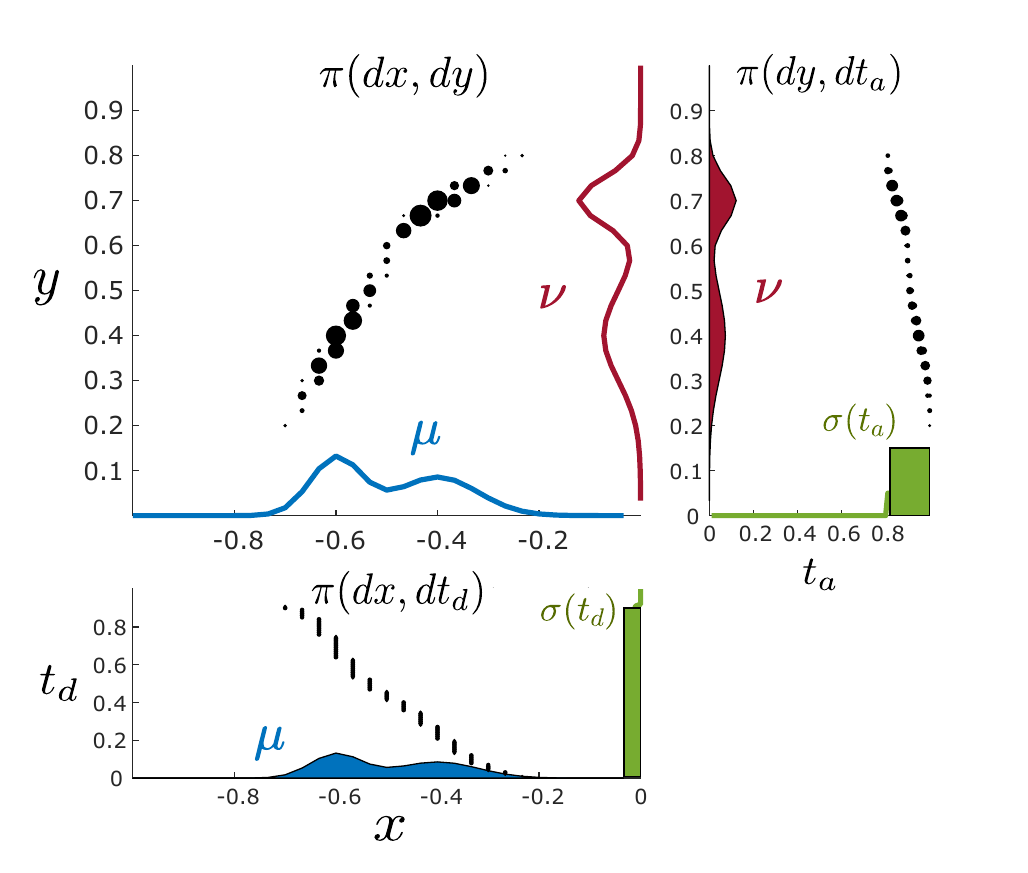}
\caption{\blue 
In earlier formulations, all particles/agents departed at the same time $t=0$ and arrived at the same time $t=t_f$.
Here, besides
meeting flow-rate constraints, 
we stratify departure and arrival so that particles/agents closer to the toll depart first, and arrive at the most distant target location again first (as it would be natural for a convoy transferring goods).
Departure and arrival rate bounds are set to $r_d = 1.1$ and $r_a = 5$, respectively. Coupling measures $\pi(dx,dy)$, $\pi(dx,dt_d)$, and $\pi(dy,dt_a)$ between timing variable are computed; it is seen that these are supported on graphs of maps (showing a monotonic correspondence).
For instance, $\pi(dx,dt_d)$ couples $\mu(dx)$ and $\sigma(t_d)$, with $\sigma(t_d) \leq r_d$ being the departure-time marginal. Due to the monotonicity of the cost $c(x,t_d) = t_d/x^2$, the coupling $\pi(dx,dt_d)$ indicates that mass closer to the toll departs earlier, while abiding by the departure rate bound $r_d$. For the timing of arrival $t_a$, properties of the coupling $\pi(dy,dt_a)$ are completely analogous.}\label{fig:DepArrGaussian}
\end{figure}

{\blue In closing, we mention that distributed flux constraints may fruitfully capture properties of matter through which transport takes place, e.g., in transport of pollutants through porous media. A generalization of the framework herein to a distributed setting would be desirable and at present open.}

\section*{Acknowledgments}
\noindent
Research supported in part by the  AFOSR under grant FA9550-23-1-0096 and ARO under W911NF-22-1-0292.

\spacingset{.9}
\bibliographystyle{plain}  
\bibliography{ref}

\end{document}